\documentclass[11pt]{article}

\usepackage{amsmath} 
\usepackage{amsthm} 
\usepackage{thmtools}
\usepackage{amssymb}	
\usepackage{graphicx} 
\usepackage{multicol} 
\usepackage{multirow}
\usepackage{color}
\usepackage{bm}
\usepackage{bbm}
\usepackage[dvips,letterpaper,margin=1in,bottom=1in]{geometry}

\usepackage[utf8]{inputenc}
\usepackage[english]{babel}

\usepackage{mathtools}

\newtheorem{theorem}{Theorem}[section]

\newtheorem{lemma}[theorem]{Lemma}
\newtheorem{corollary}[theorem]{Corollary}
\newtheorem{proposition}[theorem]{Proposition}

\newtheorem{definition}{Definition}[section]

\newtheorem{remark}{Remark}[section]

\newcommand{\braket}[2]{\left< #1 \vphantom{#2} \middle| #2 \vphantom{#1} \right>} 
\newcommand{\ketbra}[2]{\ensuremath{\ket{#1}\!\bra{#2}}}
\newcommand{\Dens}{\mathsf{D}}
\newcommand{\Id}{\mathbbm{1}}

\DeclarePairedDelimiter\abs{\lvert}{\rvert}

\DeclarePairedDelimiter\ket{\lvert}{\rangle}
\DeclarePairedDelimiter\bra{\langle}{\rvert}

\newcommand{\diag} {\operatorname{diag}}

\newcommand{\calE}{\mathcal{E}}

\newcommand{\calH}{\mathcal{H}}

\newcommand{\calK}{\mathcal{K}}

\usepackage{algorithm}
\usepackage{algpseudocode}

\usepackage{tabularx}
\usepackage{booktabs}
\usepackage{threeparttable}
\usepackage{adjustbox} 

\newcommand{\footremember}[2]{%
    \footnote{#2}
    \newcounter{#1}
    \setcounter{#1}{\value{footnote}}%
}

\usepackage{tikz}
\usetikzlibrary{quantikz2}
\usepackage[
    pagebackref,
    colorlinks,
    linkcolor=magenta,
    citecolor=magenta
]{hyperref}
\usepackage[capitalize,noabbrev]{cleveref}

\title{No-Go Theorems for Quantum Transport Metrics\\ from Fixed Cost Operators}
\author{Minbo Gao\footremember{1}{Institute of Software, CAS. \href{mailto:gaomb@ios.ac.cn}{\nolinkurl{gaomb@ios.ac.cn}} or
\href{mailto:gmb17@tsinghua.org.cn}{\nolinkurl{gmb17@tsinghua.org.cn}}.}
\and Zhengfeng Ji\footremember{2}{Tsinghua University.
\href{mailto:jizhengfeng@tsinghua.edu.cn}{\nolinkurl{jizhengfeng@tsinghua.edu.cn}}.}
\and Tianshi Yu\footremember{3}{Institute of Software, CAS. \href{mailto:yuts@ios.ac.cn}{\nolinkurl{yuts@ios.ac.cn}}}
}

\begin{document}

\maketitle

\begin{abstract}
Coupling-based quantum optimal transport generalizes classical optimal
transport by representing transport plans as bipartite states with prescribed
marginals and evaluating their cost as the expectation of a fixed Hermitian
operator.  Friedland et al. conjectured that the square root of the optimal
cost associated with the SWAP projector is a metric in every dimension and
that this property persists for nearby quantum cost matrices~\cite{FECZ22}.
Miller disproved both conjectures by constructing explicit diagonal qutrit
counterexamples~\cite{Mil26}.  Building on his analysis, we prove a uniform
no-go theorem for standard couplings.  In every dimension $d\geq3$, no fixed
cost operator makes either the optimal cost or its square root a metric, with
violations occurring already among commuting states.  The obstruction
persists under stabilization of the SWAP cost.  For channel-induced
couplings, global nonnegativity and vanishing self-cost force the cost
operator to be zero, precluding point separation when $d\geq2$.  Taken
together, these no-go results show that fixed-cost coupling formulations do
not lead to metrics on the full quantum state space.
\end{abstract}

\section{Introduction}

Classical optimal transport asks how to move one probability distribution
to another at minimum cost.  For probability measures $\mu,\nu$ on a Polish
metric space $(\mathsf{X},d)$ and a measurable cost
$c:\mathsf{X}\times\mathsf{X}\to[0,\infty]$, the Monge problem minimizes
$\int c(x,T(x))\,\mathrm{d}\mu(x)$ over measurable maps $T$ such that
$T_\#\mu=\nu$, equivalently,
$\int h(T(x))\,\mathrm{d}\mu(x)=\int h(y)\,\mathrm{d}\nu(y)$ for every
bounded Borel $h$~\cite[\S~1.1, p.~2]{Vil09}.  Such a map need not exist:
for example, on $\mathbb{R}$, if $\mu=\delta_0$ and
$\nu=\tfrac12\delta_{-1}+\tfrac12\delta_1$, then
$T_\#\mu=\delta_{T(0)}$ for every map $T$, so the Monge feasible set is
empty~\cite[\S~1.1, p.~2]{Vil09}.

Kantorovich's 1942 formulation instead optimizes over transport plans,
namely probability measures on $\mathsf{X}\times\mathsf{X}$ with prescribed
marginals~\cite{Kan58}:
\begin{equation}\label{eq:intro-classical-couplings}
    \Pi(\mu,\nu)
    :=
    \left\{
        \pi\in\mathcal{P}(\mathsf{X}\times\mathsf{X}):
        \pi\text{ has marginals }\mu\text{ and }\nu
    \right\}.
\end{equation}
For a ground cost $c$, this gives
$\mathcal{T}_c(\mu,\nu):=\inf_{\pi\in\Pi(\mu,\nu)}
\int c\,\mathrm{d}\pi$.  For $p\geq1$, let
$\mathcal{P}_p(\mathsf{X})$ denote the Borel probability measures on
$\mathsf{X}$ with finite moment of order $p$.  For
$\mu,\nu\in\mathcal{P}_p(\mathsf{X})$, choosing $c(x,y)=d(x,y)^p$ yields a distance called
Wasserstein distance of order $p$:
\begin{equation}\label{eq:intro-wasserstein}
    W_p(\mu,\nu)
    :=
    \left(\inf_{\pi\in\Pi(\mu,\nu)}
    \int d(x,y)^p\,\mathrm{d}\pi(x,y)\right)^{1/p}.
\end{equation}
This is a distance on $\mathcal{P}_p(\mathsf{X})$; see
\cite[pp.~94--95]{Vil09}.  Wasserstein geometry has become a standard
tool in the study of gradient flows and nonlinear PDEs~\cite{Amb08},
statistics~\cite{PZ19}, and computational optimal transport and data
science~\cite{PC19}.

A coupling-based quantum analogue~\cite{FECZ22,CEFZ23} replaces measures
by density operators, couplings by bipartite states with prescribed partial
traces, and expected costs by expectations of fixed operators.  For states
$\rho,\sigma$ on a finite-dimensional Hilbert space $\calH$, set
\begin{align}
    \Gamma(\rho,\sigma)
    &:=
    \{\omega\succeq0:
      \operatorname{Tr}_B\omega=\rho,\
      \operatorname{Tr}_A\omega=\sigma\},
    \label{eq:intro-standard-couplings}\\
    T_C(\rho,\sigma)
    &:=
    \min_{\omega\in\Gamma(\rho,\sigma)}
    \operatorname{Tr}(C\omega),
    \label{eq:intro-standard-cost}
\end{align}
for a fixed Hermitian cost $C$.  We call $\Gamma(\rho,\sigma)$ the
\emph{standard coupling set} to distinguish this direct fixed-marginal
construction from the channel-induced coupling set considered below.  Such
costs have applications in quantum
program logic~\cite{BGWZ25,BGK+26}, quantum resource
theory~\cite{Shi_2024}, quantum machine learning~\cite{CYL+19},
and semiclassical and mean-field
limits~\cite{GMP16}.
For $C\succeq0$, one may set $W_C:=\sqrt{T_C}$; the main example is the
SWAP cost $P_-=(\mathbbm{1}-F)/2$~\cite{FECZ22}.
Although the admissible costs are characterized in~\cite{CEFZ23}, the SWAP
choice already fails the triangle inequality for commuting qutrit
states, shown recently in~\cite{Mil26}.

\begin{samepage}
This motivates our main question:
\begin{quote}
\emph{Can any fixed, state-independent
linear cost define a metric on all mixed states through the above formulation of
optimization over quantum couplings?}
\end{quote}
\end{samepage}

Our answer is negative: for standard couplings, every admissible fixed cost
violates the triangle inequality for $d\geq3$, already on commuting qutrit
states.

We also study the channel-induced coupling set
\begin{equation}\label{eq:intro-channel-couplings}
    \widetilde{\Gamma}(\rho,\sigma)
    :=
    \{\omega\succeq0:
      \operatorname{Tr}_B\omega=\rho^\top,\
      \operatorname{Tr}_A\omega=\sigma\},
\end{equation}
where the transpose is a Choi-basis convention.
For channel-induced couplings, nonnegativity and zero self-cost
force $C=0$, so point separation fails for $d\geq2$.

Thus, for $d\geq3$, a quantum Wasserstein distance on all mixed states
cannot arise from a direct quantum analogue of the classical fixed-cost
coupling linear program; finding the right replacement remains open.

\subsection{Main results}

For a set $\mathsf{S}$, we call a map
$D:\mathsf{S}\times\mathsf{S}\to[0,\infty)$ a \emph{semidistance} if
it is symmetric and vanishes exactly on the diagonal; the triangle
inequality is not required.

Our first result answers the open question for the standard coupling
set.

\begin{theorem}[Standard-coupling no-go]\label{thm:intro-standard-no-go}
    Let $d\geq 3$ and let $C=C^\dagger$ act on
    $\mathbb{C}^d\otimes\mathbb{C}^d$.  Suppose that $T_C$ in
    \eqref{eq:intro-standard-cost} is a semidistance on the density operators
    on $\mathbb{C}^d$.  Then there are three mutually commuting states
    $\rho,\sigma,\tau$ such that
    \begin{align}
        T_C(\rho,\tau)
        &>
        T_C(\rho,\sigma)+T_C(\sigma,\tau),\label{eq:intro-T-violation}\\
        W_C(\rho,\tau)
        &>
        W_C(\rho,\sigma)+W_C(\sigma,\tau).
        \label{eq:intro-W-violation}
    \end{align}
    Consequently, no fixed state-independent linear cost makes $W_C$ a
    metric on the full mixed-state space when $d\geq 3$.
\end{theorem}

Theorem~\ref{thm:intro-standard-no-go} is uniform over all admissible
Hermitian costs: it is not specific to the SWAP projector or to a
perturbation of it.  Moreover, one fixed triple of commuting qutrit states
provides the obstruction after a basis change depending on $C$.  The
commutativity of the marginals does not make the optimization classical.
Even for diagonal states, an optimal quantum coupling can use coherence
between $\ket{ij}$ and $\ket{ji}$, and this coherence is precisely what
drives the violation.  The threshold is sharp within this framework: in
dimension two, admissibility forces $C$ to be a positive multiple of $P_-$,
and the corresponding $W_C$ is a metric~\cite{CEFZ23}.

We also consider the stabilization of the SWAP cost introduced
by M\"uller-Hermes~\cite{Mul22}.  Write $T:=T_{P_-}$ and
\begin{equation}\label{eq:intro-stabilized}
    T_s(\rho,\sigma)
    :=
    \inf_{\gamma}
    T(\rho\otimes\gamma,\sigma\otimes\gamma),
    \qquad
    W_s(\rho,\sigma):=\sqrt{T_s(\rho,\sigma)},
\end{equation}
where the infimum is over finite-dimensional ancillary states.
Stabilization gives channel monotonicity and invariance under tensoring both
arguments with a common state~\cite{Mul22}.  It does not, however, restore
the triangle inequality.

\begin{corollary}[Stabilized SWAP cost]\label{cor:intro-stabilized-no-go}
    There are three mutually commuting qutrit states $\rho,\sigma,\tau$ for
    which
    \[
        W_s(\rho,\tau)
        >
        W_s(\rho,\sigma)+W_s(\sigma,\tau).
    \]
    In particular, the stabilized SWAP semidistance is not a metric.
\end{corollary}

For the channel-induced set \eqref{eq:intro-channel-couplings}, define
\begin{equation}\label{eq:intro-channel-cost}
    \widetilde{T}_C(\rho,\sigma)
    :=
    \min_{\omega\in\widetilde{\Gamma}(\rho,\sigma)}
    \operatorname{Tr}(C\omega).
\end{equation}
Here the obstruction occurs before one reaches the triangle inequality.

\begin{theorem}[Channel-induced no-go]\label{thm:intro-channel-no-go}
    Let $C=C^\dagger$ act on
    $\mathbb{C}^d\otimes\mathbb{C}^d$.  Suppose that
    $\widetilde{T}_C(\rho,\sigma)\geq 0$ for all states $\rho,\sigma$ and that
    $\widetilde{T}_C(\rho,\rho)=0$ for every state $\rho$.  Then $C=0$.
    Consequently, when $d\geq 2$, no fixed linear cost over channel-induced
    couplings can define a semidistance, and hence none can define a metric.
\end{theorem}

The two theorems identify different failure modes.  For the standard
coupling set, zero self-cost leaves a nontrivial class of admissible
operators, but none satisfies the triangle inequality once $d\geq 3$.  For
the channel-induced set, zero self-cost is already incompatible with point
separation once $d\geq 2$.


\subsection{Technique overview}

The proof of Theorem~\ref{thm:intro-standard-no-go} begins with the
characterization of admissible costs from~\cite{CEFZ23}.  If $T_C$ is a
semidistance, then
\begin{equation}\label{eq:intro-admissible}
    C\succeq 0,
    \qquad
    \ker C=\operatorname{Sym}^2(\mathbb{C}^d).
\end{equation}
Thus $C$ vanishes on the symmetric subspace and is positive definite on the
antisymmetric subspace.  After compressing to an arbitrary
three-dimensional subspace, we use the fact that every orthonormal basis of
$\wedge^2\mathbb{C}^3$ is, up to phases, the wedge basis induced by some
one-particle basis.  Thus the restriction of $C$ to the antisymmetric sector
can be diagonalized by a common basis change $U\otimes U$.  Up to relabeling,
every admissible qutrit cost therefore reduces to
\begin{equation}\label{eq:intro-weighted-cost}
    C_{a,b,c}
    =
    a\ketbra{1\wedge 2}{1\wedge 2}
    +b\ketbra{1\wedge 3}{1\wedge 3}
    +c\ketbra{2\wedge 3}{2\wedge 3},
    \qquad a,b,c>0.
\end{equation}

For diagonal states, \cite[Proposition~1]{Mil26} gives the exact reduction
\begin{equation}\label{eq:intro-diagonal-formula}
    \begin{aligned}
    T_{C_{a,b,c}}(\diag p,\diag q)
    &=
    \frac12
    \min_{\gamma\in\Pi(p,q)}
    \left[
    a\left(\sqrt{\gamma_{12}}-\sqrt{\gamma_{21}}\right)^2
    +b\left(\sqrt{\gamma_{13}}-\sqrt{\gamma_{31}}\right)^2
    \right.\\[-0.2em]
    &\hspace{9em}\left.
    +c\left(\sqrt{\gamma_{23}}-\sqrt{\gamma_{32}}\right)^2
    \right].
    \end{aligned}
\end{equation}
The diagonal entries of the quantum transport plan form a classical
coupling $\gamma$, but positivity allows each
$\{\ket{ij},\ket{ji}\}$ block to saturate its coherence bound.  This produces
the square-root interference terms in \eqref{eq:intro-diagonal-formula}.
We then exhibit one triple of diagonal qutrit states and
evaluate its three pairwise costs.  After placing the smallest of $a,b,c$ in
one fixed antisymmetric direction, the same triple strictly satisfies
\eqref{eq:intro-T-violation} and \eqref{eq:intro-W-violation} for every
positive choice of the remaining weights.  This counterexample extends the
SWAP-specific obstruction in~\cite{Mil26} to the general no-go result of
Theorem~\ref{thm:intro-standard-no-go}.

For Corollary~\ref{cor:intro-stabilized-no-go}, we show that a
common ancillary state does not change the SWAP transport cost between
commuting states.  After simultaneous diagonalization, the cost can be
expressed through
\[
    \Phi(\pi):=\sum_{i,j}\sqrt{\pi_{ij}\pi_{ji}}
\]
for a classical coupling $\pi$.  Coarse-graining a coupling of
$p\otimes r$ and $q\otimes r$ cannot increase $\Phi$, by the
Cauchy--Schwarz inequality, while tensoring an optimal coupling of $p$ and
$q$ with the diagonal ancillary state attains the reverse bound.  Hence
$T_s=T$ on commuting pairs, and Miller's qutrit violation transfers
unchanged to the stabilized setting.

The proof of Theorem~\ref{thm:intro-channel-no-go} is shorter and uses a
different approach.  Every bipartite state $\omega$ is feasible in
\eqref{eq:intro-channel-couplings} for
\[
    \rho=(\operatorname{Tr}_B\omega)^\top,
    \qquad
    \sigma=\operatorname{Tr}_A\omega.
\]
Nonnegativity of all the minima in
\eqref{eq:intro-channel-cost} therefore implies
$\operatorname{Tr}(C\omega)\geq0$ for every bipartite state, and hence
$C\succeq0$.  If $\rho=\ketbra{\psi}{\psi}$ is pure, its only
channel-induced self-coupling is the pure product state on
$\overline{\psi}\otimes\psi$.  Zero self-cost and positivity force
\[
    C\ket{\overline{\psi}\otimes\psi}=0
    \quad\text{for every }\psi.
\]
These vectors span $\mathbb{C}^d\otimes\mathbb{C}^d$, so $C=0$.  No
triangle-inequality argument is needed.

\subsection{Discussion}

\paragraph{Scope of the standard-coupling no-go theorem.}
The dimension threshold $d\geq3$ is sharp.  In dimension two, admissibility
forces the cost to be a positive multiple of $P_-$, and the corresponding
square-root cost is a metric on the full state space~\cite{FECZ22,CEFZ23}.
The theorem is also specific to the full mixed-state space and does not rule
out metric behavior on restricted families.  For example, on pure states the
SWAP cost reduces to
$W_{P_-}=2^{-1/2}\sqrt{1-|\braket{\psi}{\phi}|^2}$, a constant multiple of
root infidelity and hence a metric~\cite[Proposition~1]{FECZ22}.

\paragraph{Failure of quantum gluing.}
Our standard-coupling no-go theorem identifies the triangle inequality as
the missing metric axiom, so it is useful to recall why that axiom holds
classically.  Choose nearly optimal couplings
$\pi_{12}\in\Pi(\mu_1,\mu_2)$ and $\pi_{23}\in\Pi(\mu_2,\mu_3)$.  The Gluing
Lemma~\cite[p.~23]{Vil09} produces a joint law $\pi_{123}$ with these
two pairwise marginals, whose $(1,3)$-marginal is therefore a coupling of
$\mu_1$ and $\mu_3$.  The ground-space triangle inequality, followed by
Minkowski's inequality under $\pi_{123}$, then gives
$W_p(\mu_1,\mu_3)\leq W_p(\mu_1,\mu_2)+W_p(\mu_2,\mu_3)$.  No analogous
principle holds for general quantum couplings: two bipartite states with a
common marginal need not admit a common tripartite extension.  Our theorem
shows that the missing gluing lemma is not merely a gap in the classical
proof; in the fixed-cost, one-shot framework, the triangle inequality
itself fails.

\paragraph{No fixed-cost coupling representation of trace distance.}
On a finite space with the discrete metric, $W_1$ is the total variation
distance.  Its quantum counterpart is the trace distance
$d_{\mathrm{tr}}(\rho,\sigma)=\frac12\|\rho-\sigma\|_1$.
Zhou, Yu, Ying, and Ying proved that no
Hermitian cost $H$ and bijection $f$ can satisfy
$d_{\mathrm{tr}}(\rho,\sigma)=f(T_H(\rho,\sigma))$ for all quantum
states~\cite[Theorem~7]{ZYYY22}.  Their theorem excludes recovering this
particular metric even after a reparametrization of the coupling cost,
whereas our result excludes metricity of $\sqrt{T_C}$ for every admissible
fixed cost when $d\geq3$.

\section{Preliminaries}\label{sec:preliminaries}

\subsection{Notation and coupling costs}

Let $\calH\cong\mathbb{C}^d$ be a finite-dimensional Hilbert space, and
let $\Dens(\calH)$ denote its density operators.  We fix an orthonormal basis
$\{\ket{1},\ldots,\ket{d}\}$ whenever a transpose, entrywise conjugate, or
Choi operator is used.  Subscripts on partial traces name the subsystem being
traced out.  The standard coupling set $\Gamma(\rho,\sigma)$ and its fixed-cost
value $T_C(\rho,\sigma)$ are as in
\eqref{eq:intro-standard-couplings} and \eqref{eq:intro-standard-cost}.

Let $F$ be the SWAP operator on $\calH\otimes\calH$, defined by
$F(x\otimes y)=y\otimes x$, and write
\begin{equation}\label{eq:symmetric-decomposition}
    \calH_{\mathrm{sym}}:=\operatorname{Sym}^2(\calH),
    \qquad
    \calH_{\mathrm{asym}}:=\wedge^2\calH,
    \qquad
    P_-:=\frac{\Id-F}{2}.
\end{equation}
For $\ket{x},\ket{y}\in\calH$, write
\begin{equation}\label{eq:wedge-basis}
    \ket{x\wedge y}
    :=\frac{\ket{x}\otimes\ket{y}-\ket{y}\otimes\ket{x}}{\sqrt{2}}.
\end{equation}
In particular, the vectors $\ket{i\wedge j}$ with $i<j$ form an
orthonormal basis of $\calH_{\mathrm{asym}}$.
Throughout, ``semidistance'' has the convention stated in the introduction:
it includes nonnegativity, symmetry, and point separation, but not the
triangle inequality.

\subsection{Channels and Choi-induced couplings}

Following De Palma and Trevisan~\cite{DT21}, we regard a quantum channel as
the quantum analogue of the stochastic map associated with a classical
transport plan.  For states $\rho,\sigma\in\Dens(\calH)$, let
\begin{equation}\label{eq:quantum-transport-plans}
    \mathsf{M}(\rho,\sigma)
    :=
    \left\{
        \calE:\mathcal{L}(\operatorname{supp}\rho)
        \to\mathcal{L}(\calH):
        \calE\text{ is a quantum channel and }\calE(\rho)=\sigma
    \right\}.
\end{equation}
The restriction of the input space to $\operatorname{supp}\rho$ is essential
for the correspondence below to be one-to-one when $\rho$ is not full rank.

To pass from a transport plan to a bipartite state, let
\[
    \ket{\Omega}:=\sum_{i=1}^d\ket{ii}
\]
be the unnormalized maximally entangled vector and define the canonical
purification
\begin{equation}\label{eq:canonical-purification}
    \ket{\Omega_\rho}
    :=(\sqrt{\rho^{\top}}\otimes\Id)\ket{\Omega}
    =(\Id\otimes\sqrt\rho)\ket{\Omega}.
\end{equation}
Its two marginals are
\[
    \operatorname{Tr}_B\ketbra{\Omega_\rho}{\Omega_\rho}=\rho^{\top},
    \qquad
    \operatorname{Tr}_A\ketbra{\Omega_\rho}{\Omega_\rho}=\rho.
\]
For $\calE\in\mathsf{M}(\rho,\sigma)$, set
\begin{equation}\label{eq:weighted-choi-state}
    \omega_{\rho,\calE}
    :=(\operatorname{id}\otimes\calE)
    \bigl(\ketbra{\Omega_\rho}{\Omega_\rho}\bigr).
\end{equation}
The two marginals are
\begin{equation}\label{eq:weighted-choi-marginals}
    \operatorname{Tr}_B\omega_{\rho,\calE}=\rho^{\top},
    \qquad
    \operatorname{Tr}_A\omega_{\rho,\calE}=\calE(\rho)=\sigma.
\end{equation}
This explains both the transpose in
\eqref{eq:intro-channel-couplings} and the inclusion
$\omega_{\rho,\calE}\in\widetilde\Gamma(\rho,\sigma)$.

\begin{proposition}[Channel--coupling correspondence~{\cite[Proposition~2]{DT21}}]
\label{prop:weighted-choi-correspondence}
    The map $\calE\mapsto\omega_{\rho,\calE}$ is a bijection from
    $\mathsf{M}(\rho,\sigma)$ onto
    $\widetilde\Gamma(\rho,\sigma)$.
\end{proposition}


For completeness, a transport plan on $\operatorname{supp}\rho$ can always
be extended to a channel on all of $\calH$.  For any such extension, let
\[
    J(\calE):=(\operatorname{id}\otimes\calE)
    \bigl(\ketbra{\Omega}{\Omega}\bigr)
\]
be its ordinary Choi operator, so that
$\operatorname{Tr}_B J(\calE)=\Id$.  Equation
\eqref{eq:weighted-choi-state} can then equivalently be written as
\begin{equation}\label{eq:weighted-choi-operator}
    \omega_{\rho,\calE}
    =
    (\sqrt{\rho^{\top}}\otimes\Id)
    J(\calE)
    (\sqrt{\rho^{\top}}\otimes\Id),
\end{equation}
which is independent of the chosen extension.  This is why we call
$\widetilde\Gamma(\rho,\sigma)$ the Choi-induced coupling set.

De Palma and Trevisan define the transport cost from self-adjoint observables
$R_1,\ldots,R_N$~\cite[Definition~5 and Eq.~(32)]{DT21}.  After swapping the
tensor factors to match our convention, set
\[
    D_i:=R_i^{\top}\otimes\Id-\Id\otimes R_i,
    \qquad
    C_{\mathbf R}:=\sum_{i=1}^N D_i^2.
\]
Their cost of a channel-induced coupling $\omega$ is then
\[
    \mathcal{C}_{\mathbf R}(\omega)
    :=\sum_{i=1}^N\operatorname{Tr}(D_i\omega D_i)
    =\operatorname{Tr}(C_{\mathbf R}\omega).
\]
Proposition~\ref{prop:weighted-choi-correspondence} therefore identifies
optimization over transport channels with
\[
    \min_{\calE\in\mathsf{M}(\rho,\sigma)}
    \mathcal{C}_{\mathbf R}(\omega_{\rho,\calE})
    =
    \min_{\omega\in\widetilde\Gamma(\rho,\sigma)}
    \operatorname{Tr}(C_{\mathbf R}\omega).
\]
We instead allow any Hermitian cost $C$ and use $\widetilde T_C$ from
\eqref{eq:intro-channel-cost}.

\section{Fixed-cost no-go for standard couplings}
\label{sec:standard-no-go}

The proof of Theorem~\ref{thm:intro-standard-no-go} first reduces an
arbitrary cost whose optimal value is a semidistance to three positive
weights on a qutrit, then solves the resulting optimization for diagonal
states.

\subsection{Structural reduction of admissible costs}

For simplicity, we use the following definition.
\begin{definition}[Admissible cost]\label{def:admissible-cost}
    We call a Hermitian cost $C$ on $\calH\otimes\calH$ admissible if $T_C$
    is a semidistance on $\Dens(\calH)$.
\end{definition}

In~\cite{CEFZ23}, they give a characterization for admissible cost operators, which we recall as follows.

\begin{theorem}[Admissible-cost characterization~{\cite[Theorem~6.3]{CEFZ23}}]\label{thm:admissible-costs}
    A Hermitian cost $C$ is admissible if and only if
    \begin{equation}\label{eq:admissible-cost}
        C\succeq0,
        \qquad
        \ker C=\calH_{\mathrm{sym}}.
    \end{equation}
\end{theorem}

In particular, an admissible cost vanishes on the entire symmetric subspace
and is positive definite on the antisymmetric subspace.  This is stronger
than requiring zero expectation only on pure tensors $\ket{\psi}\otimes \ket{\psi}$.

\begin{lemma}[Restriction to a common support]\label{lem:common-support}
    Let $\rho$ and $\sigma$ be supported on $\calK\subseteq\calH$, and let
    $P_{\calK}$ be the orthogonal projection onto $\calK$.  Define the
    compression of $C$ to $\calK\otimes\calK$ by
    \[
        C_{\calK}
        :=(P_{\calK}\otimes P_{\calK})
          C
          (P_{\calK}\otimes P_{\calK}),
    \]
    viewed as an operator on $\calK\otimes\calK$.
    Then
    \[
        T_C(\rho,\sigma)=T_{C_{\calK}}(\rho,\sigma).
    \]
    If $C$ is admissible on $\calH$, then $C_{\calK}$ is admissible on
    $\calK$.
\end{lemma}

\begin{proof}
    Set $Q_{\calK}:=\Id-P_{\calK}$ and let
    $\omega\in\Gamma(\rho,\sigma)$.  Since $\rho$ is supported on $\calK$,
    \[
        \operatorname{Tr}\bigl((Q_{\calK}\otimes\Id)\omega\bigr)
        =\operatorname{Tr}\bigl(Q_{\calK}\operatorname{Tr}_B\omega\bigr)
        =\operatorname{Tr}(Q_{\calK}\rho)=0.
    \]
    Positivity of $\omega$ then gives
    \[
        0
        =\operatorname{Tr}\bigl(
          \omega^{1/2}(Q_{\calK}\otimes\Id)\omega^{1/2}
          \bigr)
        =\bigl\|(Q_{\calK}\otimes\Id)\omega^{1/2}\bigr\|_2^2,
    \]
    and hence $(Q_{\calK}\otimes\Id)\omega=0$; taking adjoints also gives
    $\omega(Q_{\calK}\otimes\Id)=0$.  Applying the same argument to the
    second marginal gives the corresponding two identities with
    $\Id\otimes Q_{\calK}$.  Therefore
    \[
        \omega
        =(P_{\calK}\otimes P_{\calK})\omega
         (P_{\calK}\otimes P_{\calK}),
    \]
    so every feasible coupling is supported on $\calK\otimes\calK$.  By
    cyclicity of the trace,
    \[
        \operatorname{Tr}(C\omega)
        =\operatorname{Tr}(C_{\calK}\omega).
    \]
    Taking the minimum over $\Gamma(\rho,\sigma)$ proves the cost identity.

    If $C$ is admissible, then \eqref{eq:admissible-cost} holds.  Since
    $\langle x,C_{\calK}x\rangle=\langle x,Cx\rangle$ for every
    $x\in\calK\otimes\calK$, we have $C_{\calK}\succeq0$ and
    \[
        \ker C_{\calK}
        =(\calK\otimes\calK)\cap\ker C
        =(\calK\otimes\calK)\cap\operatorname{Sym}^2(\calH)
        =\operatorname{Sym}^2(\calK).
    \]
    Theorem~\ref{thm:admissible-costs} now implies that $C_{\calK}$ is
    admissible on $\calK$.
\end{proof}

The three-dimensional antisymmetric representation permits a further
normal-form reduction that is special to qutrits.

\begin{lemma}[Qutrit antisymmetric normal form]\label{lem:qutrit-normal-form}
    Let $\calK\cong\mathbb{C}^3$, and let $C\succeq0$ satisfy
    $\ker C=\operatorname{Sym}^2(\calK)$.  There exist a unitary $U$ on
    $\calK$ and positive numbers $a,b,c$ such that
    \begin{equation}\label{eq:qutrit-normal-form}
        (U\otimes U)^\dagger C(U\otimes U)
        =a\ketbra{1\wedge2}{1\wedge2}
         +b\ketbra{1\wedge3}{1\wedge3}
         +c\ketbra{2\wedge3}{2\wedge3}.
    \end{equation}
    After a permutation of the basis, one may additionally impose
    $b=\min\{a,b,c\}$.
\end{lemma}

\begin{proof}
    Since $C$ is self-adjoint,
    \[
        \operatorname{ran}C=(\ker C)^\perp=\wedge^2\calK.
    \]
    Thus $C$ vanishes on $\operatorname{Sym}^2(\calK)$ and restricts to a
    positive-definite operator on $\wedge^2\calK$.  The spectral theorem
    therefore provides positive numbers
    $\lambda_1,\lambda_2,\lambda_3$ and an orthonormal basis
    $\ket{\phi_1},\ket{\phi_2},\ket{\phi_3}$ of $\wedge^2\calK$ such that
    \[
        C=\sum_{i=1}^3\lambda_i\ketbra{\phi_i}{\phi_i}.
    \]

    Define $J:\calK\to\wedge^2\calK$ on the fixed basis by
    \[
        J\ket1=\ket{2\wedge3},
        \qquad
        J\ket2=-\ket{1\wedge3},
        \qquad
        J\ket3=\ket{1\wedge2},
    \]
    and extend it conjugate-linearly, meaning that
    \[
        J(\alpha x+\beta y)
        =\overline\alpha Jx+\overline\beta Jy.
    \]
    The vectors $\ket{2\wedge3}$, $-\ket{1\wedge3}$, and
    $\ket{1\wedge2}$ form an orthonormal basis of $\wedge^2\calK$.  Thus
    $J$ is bijective and antiunitary: for all $x,y\in\calK$,
    \[
        \langle Jx,Jy\rangle=\overline{\langle x,y\rangle},
        \qquad
        \|Jx\|=\|x\|.
    \]
    Set $\ket{u_i}:=J^{-1}\ket{\phi_i}$.  Since $J^{-1}$ is also
    antiunitary, $\ket{u_1},\ket{u_2},\ket{u_3}$ form an orthonormal basis
    of $\calK$.  Let $U$ be the unitary determined by
    $U\ket i=\ket{u_i}$, and write
    \[
        \ket{u_i}=\sum_{r=1}^3 U_{ri}\ket r.
    \]
    We now verify explicitly how $J$ intertwines $U$ with its second
    exterior power.  Bilinearity of the wedge product gives
    \begin{align*}
        \ket{u_2\wedge u_3}
        ={}&(U_{22}U_{33}-U_{32}U_{23})\ket{2\wedge3}\\
          &+(U_{12}U_{33}-U_{32}U_{13})\ket{1\wedge3}\\
          &+(U_{12}U_{23}-U_{22}U_{13})\ket{1\wedge2}.
    \end{align*}
    In terms of the matrix cofactors, the three coefficients on the
    right-hand side are, respectively,
    \[
        \operatorname{cof}_{11}(U),\qquad
        -\operatorname{cof}_{21}(U),\qquad
        \operatorname{cof}_{31}(U).
    \]
    The adjugate identity and unitarity of $U$ imply
    \[
        \operatorname{cof}_{r1}(U)
        =(\det U)(U^{-1})_{1r}
        =(\det U)\overline{U_{r1}}.
    \]
    Consequently,
    \[
        \ket{u_2\wedge u_3}
        =(\det U)\bigl(
            \overline{U_{11}}\ket{2\wedge3}
           -\overline{U_{21}}\ket{1\wedge3}
           +\overline{U_{31}}\ket{1\wedge2}
          \bigr)
        =(\det U)J\ket{u_1}.
    \]
    The two cyclic analogues are
    \[
        \ket{u_3\wedge u_1}=(\det U)J\ket{u_2},
        \qquad
        \ket{u_1\wedge u_2}=(\det U)J\ket{u_3}.
    \]
    Since $J\ket{u_i}=\ket{\phi_i}$ and
    $\ket{u_3\wedge u_1}=-\ket{u_1\wedge u_3}$, we obtain
    \[
        \ket{\phi_1}=(\det U)^{-1}\ket{u_2\wedge u_3},
        \qquad
        \ket{\phi_2}=-(\det U)^{-1}\ket{u_1\wedge u_3},
        \qquad
        \ket{\phi_3}=(\det U)^{-1}\ket{u_1\wedge u_2}.
    \]
    Because $\abs{\det U}=1$, the phase factors disappear upon passing to
    rank-one projectors.  Substituting into the spectral decomposition of
    $C$ yields
    \[
        C
        =\lambda_1\ketbra{u_2\wedge u_3}{u_2\wedge u_3}
         +\lambda_2\ketbra{u_1\wedge u_3}{u_1\wedge u_3}
         +\lambda_3\ketbra{u_1\wedge u_2}{u_1\wedge u_2}.
    \]
    On the other hand, the definition of the wedge product gives
    \[
        (U\otimes U)\ket{i\wedge j}=\ket{u_i\wedge u_j}.
    \]
    Conjugating the preceding spectral decomposition by $U\otimes U$
    therefore gives
    \[
        (U\otimes U)^\dagger C(U\otimes U)
        =\lambda_3\ketbra{1\wedge2}{1\wedge2}
         +\lambda_2\ketbra{1\wedge3}{1\wedge3}
         +\lambda_1\ketbra{2\wedge3}{2\wedge3}.
    \]
    Thus \eqref{eq:qutrit-normal-form} holds with
    $(a,b,c)=(\lambda_3,\lambda_2,\lambda_1)$.  Finally, a permutation of
    the one-particle basis induces the corresponding permutation of the
    three unordered pairs $\{1,2\},\{1,3\},\{2,3\}$.  We may therefore
    send whichever pair carries the smallest coefficient to $\{1,3\}$,
    which gives $b=\min\{a,b,c\}$.
\end{proof}

\subsection{Exact reduction for diagonal states}

For probability vectors $p,q\in\mathbb{R}_+^n$, define their classical
coupling polytope by
\begin{equation}\label{eq:classical-coupling-polytope}
    \Gamma_{\mathrm{cl}}(p,q)
    :=
    \left\{
        \gamma=(\gamma_{ij})_{i,j=1}^n\geq0:
        \sum_j\gamma_{ij}=p_i,\ 
        \sum_i\gamma_{ij}=q_j
    \right\}.
\end{equation}
Consider a cost that is diagonal in the antisymmetric basis,
\begin{equation}\label{eq:weighted-antisymmetric-cost}
    C_{\bm c}
    :=\sum_{1\leq i<j\leq n}
    c_{ij}\ketbra{i\wedge j}{i\wedge j},
    \qquad c_{ij}\geq0.
\end{equation}
The following is the fixed-cost specialization of Miller's general
diagonal-state formula.

\begin{lemma}[Diagonal-state reduction~{\cite[Proposition~1]{Mil26}}]
\label{lem:diagonal-state-reduction}
    For $C_{\bm c}$ as in \eqref{eq:weighted-antisymmetric-cost},
    \begin{equation}\label{eq:diagonal-state-reduction}
        T_{C_{\bm c}}(\diag p,\diag q)
        =\frac12
        \min_{\gamma\in\Gamma_{\mathrm{cl}}(p,q)}
        \sum_{i<j}c_{ij}
        \left(\sqrt{\gamma_{ij}}-\sqrt{\gamma_{ji}}\right)^2.
    \end{equation}
\end{lemma}

\begin{proof}
    For an arbitrary quantum coupling $\omega$, define its
    diagonal-unitary twirl by
    \[
        \mathcal{T}(\omega)
        :=
        \int_{\mathbb{T}^n}
        (U_{\bm\theta}\otimes U_{\bm\theta})\omega
        (U_{\bm\theta}^\dagger\otimes U_{\bm\theta}^\dagger)
        \,d\mu(\bm\theta),
        \qquad
        U_{\bm\theta}:=\sum_k e^{\mathrm{i}\theta_k}\ketbra{k}{k}.
    \]
    Here $\mathbb{T}^n=[0,2\pi)^n$ and
    $d\mu(\bm\theta)=(2\pi)^{-n}\,d\theta_1\cdots d\theta_n$ is normalized
    Haar measure; equivalently, $\theta_1,\ldots,\theta_n$ are independent
    and uniformly distributed on $[0,2\pi)$.  Since the marginals are
    diagonal and $C_{\bm c}$ commutes with
    $U_{\bm\theta}\otimes U_{\bm\theta}$, $\mathcal{T}(\omega)$ has the same
    marginals and the same expectation of $C_{\bm c}$ as $\omega$.
    The averaged state has a scalar block on each $\ket{ii}$
    and, for every $i<j$, a block on
    $\operatorname{span}\{\ket{ij},\ket{ji}\}$ of the form
    \[
        \begin{pmatrix}
            \gamma_{ij}&z_{ij}\\
            \overline{z_{ij}}&\gamma_{ji}
        \end{pmatrix},
        \qquad
        |z_{ij}|\leq\sqrt{\gamma_{ij}\gamma_{ji}}.
    \]
    Its diagonal entries form a classical coupling
    $\gamma\in\Gamma_{\mathrm{cl}}(p,q)$.  The contribution of the
    $(i,j)$ block is
    \begin{align*}
        \frac{c_{ij}}2
        \bigl(\gamma_{ij}+\gamma_{ji}-2\operatorname{Re}z_{ij}\bigr)
        &\geq
        \frac{c_{ij}}2
        \left(\sqrt{\gamma_{ij}}-\sqrt{\gamma_{ji}}\right)^2.
    \end{align*}
    Equality is attained in every block by taking
    $z_{ij}=\sqrt{\gamma_{ij}\gamma_{ji}}$.  These blocks, together with
    the diagonal entries $\gamma_{ii}$, form a positive coupling with the
    required marginals.  Minimizing over $\gamma$ proves
    \eqref{eq:diagonal-state-reduction}.
\end{proof}

The variables $\gamma_{ij}$ are classical, but the square-root terms in
\eqref{eq:diagonal-state-reduction} come from coherent off-diagonal entries
of the quantum coupling.  Commuting marginals therefore do not reduce this
optimization to a classical linear transport problem.

\subsection{A uniform weighted qutrit counterexample}

For positive $a,b,c$, let
\begin{equation}\label{eq:weighted-qutrit-cost}
    C_{a,b,c}
    :=a\ketbra{1\wedge2}{1\wedge2}
      +b\ketbra{1\wedge3}{1\wedge3}
      +c\ketbra{2\wedge3}{2\wedge3}.
\end{equation}
After relabeling the qutrit basis, assume that
\begin{equation}\label{eq:minimum-weight}
    0<b\leq a,
    \qquad
    b\leq c.
\end{equation}
Let
\begin{equation}\label{eq:qutrit-triple}
    \begin{aligned}
        p&=\left(\frac34,\frac14,0\right),
        &\rho&=\diag p,\\
        q&=\left(\frac14,\frac58,\frac18\right),
        &\sigma&=\diag q,\\
        r&=\left(\frac18,\frac58,\frac14\right),
        &\tau&=\diag r.
    \end{aligned}
\end{equation}
This is the same commuting triple for every choice of the three weights.

\begin{lemma}\label{lem:square-root-monotonicity}
    For $A\geq B>0$, the function
    \[
        g_{A,B}(\lambda)
        :=\left(\sqrt{A-\lambda}-\sqrt{B-\lambda}\right)^2
    \]
    is nondecreasing for $0\leq\lambda<B$.
\end{lemma}

\begin{proof}
    Direct differentiation gives
    \[
        g_{A,B}'(\lambda)
        =-2+\frac{A+B-2\lambda}
        {\sqrt{(A-\lambda)(B-\lambda)}}\geq0,
    \]
    where the last inequality is the arithmetic--geometric mean
    inequality.
\end{proof}

\begin{proposition}[Uniform weighted qutrit counterexample]
\label{prop:weighted-qutrit-counterexample}
    Under \eqref{eq:minimum-weight}, the three states in
    \eqref{eq:qutrit-triple} satisfy
    \begin{align}
        T_{C_{a,b,c}}(\rho,\sigma)
        &=\frac{(7-2\sqrt{10})a+b}{16},
        \label{eq:cost-rho-sigma}\\
        T_{C_{a,b,c}}(\rho,\tau)
        &=\frac{a+2b}{16},
        \label{eq:cost-rho-tau}\\
        T_{C_{a,b,c}}(\sigma,\tau)
        &=\frac{(3-2\sqrt2)b}{16}.
        \label{eq:cost-sigma-tau}
    \end{align}
    Both $T_{C_{a,b,c}}$ and $\sqrt{T_{C_{a,b,c}}}$ strictly violate the
    triangle inequality on this triple.
\end{proposition}

\begin{proof}
    We apply Lemma~\ref{lem:diagonal-state-reduction} to each pair.

    \smallskip
    \noindent\emph{The pair $(\rho,\sigma)$.}
    Every classical coupling of $p$ and $q$ has the form
    \begin{equation}\label{eq:coupling-rho-sigma-parametrization}
        \gamma(\lambda,\mu)=
        \begin{pmatrix}
            \lambda&\frac58-\lambda+\mu&\frac18-\mu\\
            \frac14-\lambda&\lambda-\mu&\mu\\
            0&0&0
        \end{pmatrix},
    \end{equation}
    where
    \[
        0\leq\lambda\leq\frac14,
        \qquad
        0\leq\mu\leq\min\left\{\lambda,\frac18\right\}.
    \]
    The objective in \eqref{eq:diagonal-state-reduction} is
    \[
        \frac12\left[
        a\left(\sqrt{\frac58-\lambda+\mu}
          -\sqrt{\frac14-\lambda}\right)^2
        +b\left(\frac18-\mu\right)+c\mu
        \right].
    \]
    For fixed $\lambda$, the squared difference is nondecreasing in $\mu$,
    while $(c-b)\mu\geq0$ by \eqref{eq:minimum-weight}.  The minimum occurs
    at $\mu=0$.  Lemma~\ref{lem:square-root-monotonicity} then shows that
    the remaining expression is minimized at $\lambda=0$.  An optimal
    coupling is
    \begin{equation}\label{eq:optimal-coupling-rho-sigma}
        \gamma_{\rho\sigma}=
        \begin{pmatrix}
            0&\frac58&\frac18\\
            \frac14&0&0\\
            0&0&0
        \end{pmatrix}.
    \end{equation}
    Substitution gives \eqref{eq:cost-rho-sigma}, since
    \[
        \left(\sqrt{\frac58}-\sqrt{\frac14}\right)^2
        =\frac{7-2\sqrt{10}}8.
    \]

    \smallskip
    \noindent\emph{The pair $(\rho,\tau)$.}
    Every classical coupling of $p$ and $r$ has the form
    \begin{equation}\label{eq:coupling-rho-tau-parametrization}
        \gamma(\lambda,\mu)=
        \begin{pmatrix}
            \lambda&\frac12-\lambda+\mu&\frac14-\mu\\
            \frac18-\lambda&\frac18+\lambda-\mu&\mu\\
            0&0&0
        \end{pmatrix},
    \end{equation}
    where
    \[
        0\leq\lambda\leq\frac18,
        \qquad
        0\leq\mu\leq\frac18+\lambda.
    \]
    The objective is
    \[
        \frac12\left[
        a\left(\sqrt{\frac12-\lambda+\mu}
          -\sqrt{\frac18-\lambda}\right)^2
        +b\left(\frac14-\mu\right)+c\mu
        \right].
    \]
    The same argument first gives $\mu=0$ and then $\lambda=0$.  Hence
    \begin{equation}\label{eq:optimal-coupling-rho-tau}
        \gamma_{\rho\tau}=
        \begin{pmatrix}
            0&\frac12&\frac14\\
            \frac18&\frac18&0\\
            0&0&0
        \end{pmatrix}
    \end{equation}
    is optimal, and substitution gives \eqref{eq:cost-rho-tau}.

    \smallskip
    \noindent\emph{The pair $(\sigma,\tau)$.}
    For $\gamma\in\Gamma_{\mathrm{cl}}(q,r)$, set
    \[
        x_{ij}:=
        \left(\sqrt{\gamma_{ij}}-\sqrt{\gamma_{ji}}\right)^2.
    \]
    Since $b\leq a,c$,
    \begin{equation}\label{eq:sigma-tau-weight-lower-bound}
        \frac12(ax_{12}+bx_{13}+cx_{23})
        \geq\frac b2(x_{12}+x_{13}+x_{23}).
    \end{equation}
    To bound the last expression, define the vectors formed from the square
    roots of the first row and the first column of $\gamma$:
    \[
        u:=\left(\sqrt{\gamma_{11}},\sqrt{\gamma_{12}},
                  \sqrt{\gamma_{13}}\right),
        \qquad
        v:=\left(\sqrt{\gamma_{11}},\sqrt{\gamma_{21}},
                  \sqrt{\gamma_{31}}\right).
    \]
    By the definition of $x_{ij}$,
    \[
        \|u-v\|_2^2=x_{12}+x_{13}.
    \]
    Since the row sums of $\gamma$ are $q$ and its column sums are $r$,
    \[
        \|u\|_2^2=\sum_{j=1}^3\gamma_{1j}=q_1=\frac14,
        \qquad
        \|v\|_2^2=\sum_{j=1}^3\gamma_{j1}=r_1=\frac18.
    \]
    The reverse triangle inequality therefore gives
    \begin{align*}
        x_{12}+x_{13}
        &=\|u-v\|_2^2\\
        &\geq\left(\|u\|_2-\|v\|_2\right)^2\\
        &=\left(\sqrt{\frac14}-\sqrt{\frac18}\right)^2
          =\frac{3-2\sqrt2}{8}.
    \end{align*}
    Since $x_{23}\geq0$, \eqref{eq:sigma-tau-weight-lower-bound} now yields
    \[
        \frac12(ax_{12}+bx_{13}+cx_{23})
        \geq \frac{b(3-2\sqrt2)}{16}.
    \]
    This lower bound is attained by
    \begin{equation}\label{eq:optimal-coupling-sigma-tau}
        \gamma_{\sigma\tau}=
        \begin{pmatrix}
            0&0&\frac14\\
            0&\frac58&0\\
            \frac18&0&0
        \end{pmatrix}.
    \end{equation}
    This proves \eqref{eq:cost-sigma-tau}.

    It remains to verify the two strict violations.  From
    \eqref{eq:cost-rho-sigma}--\eqref{eq:cost-sigma-tau},
    \begin{align}
        &T_{C_{a,b,c}}(\rho,\tau)
        -T_{C_{a,b,c}}(\rho,\sigma)
        -T_{C_{a,b,c}}(\sigma,\tau)\notag\\
        &\hspace{2cm}
        =\frac{(2\sqrt{10}-6)a+(2\sqrt2-2)b}{16}>0.
        \label{eq:transport-cost-positive-gap}
    \end{align}
    For the square roots, the desired inequality is equivalent to
    \begin{equation}\label{eq:square-root-target}
        \sqrt{a+2b}
        >\sqrt{\kappa a+b}+(\sqrt2-1)\sqrt b,
        \qquad
        \kappa:=7-2\sqrt{10}.
    \end{equation}
    Normalize by $b>0$, set $s=a/b$, and define
    \[
        h(s):=\sqrt{s+2}-\sqrt{\kappa s+1}.
    \]
    We have $h(0)=\sqrt2-1$.  Moreover,
    $0<\kappa<1/\sqrt2$ and
    \[
        h'(s)=\frac1{2\sqrt{s+2}}
              -\frac{\kappa}{2\sqrt{\kappa s+1}}>0,
    \]
    because
    \[
        \kappa s+1-\kappa^2(s+2)
        =\kappa(1-\kappa)s+1-2\kappa^2>0.
    \]
    Thus $h(s)>h(0)$ for every $s>0$, which is
    \eqref{eq:square-root-target}.
\end{proof}

\begin{proof}[Proof of Theorem~\ref{thm:intro-standard-no-go}]
    By Theorem~\ref{thm:admissible-costs}, the cost $C$ satisfies
    \eqref{eq:admissible-cost}.  Choose any three-dimensional subspace
    $\calK\subseteq\mathbb{C}^d$.  Lemma~\ref{lem:common-support} shows that
    the compressed cost $C_{\calK}$ is admissible on $\calK$ and computes
    all costs between states supported there.

    Apply Lemma~\ref{lem:qutrit-normal-form} to $C_{\calK}$.  In the
    resulting basis it has the form $C_{a,b,c}$ with $0<b\leq a,c$.  Embed
    the three states in \eqref{eq:qutrit-triple} into $\mathbb{C}^d$ by
    extending them by zero on $\calK^\perp$.  They remain mutually
    commuting, and Proposition~\ref{prop:weighted-qutrit-counterexample}
    gives \eqref{eq:intro-T-violation} and
    \eqref{eq:intro-W-violation}.
\end{proof}

\begin{remark}[The qubit exception]\label{rem:qubit-exception}
    When $d=2$, the antisymmetric subspace is one-dimensional, so
    Theorem~\ref{thm:admissible-costs} forces every admissible cost to have
    the form $C=\lambda P_-$ with $\lambda>0$.  In this case
    $W_C=\sqrt\lambda\,W_{P_-}$, and $W_C$ is a metric on the Bloch
    ball~\cite{FECZ22,CEFZ23}.  The dimension threshold in
    Theorem~\ref{thm:intro-standard-no-go} is therefore sharp.
\end{remark}

\section{Stabilization does not restore the triangle inequality}
\label{sec:stabilization}

For the SWAP cost, write
\[
    T(\rho,\sigma):=T_{P_-}(\rho,\sigma),
    \qquad
    W(\rho,\sigma):=\sqrt{T(\rho,\sigma)}.
\]
M\"uller-Hermes introduced the stabilized quantities
\begin{equation}\label{eq:stabilized-cost}
    T_s(\rho,\sigma)
    :=\inf_{\gamma}T(\rho\otimes\gamma,\sigma\otimes\gamma),
    \qquad
    W_s(\rho,\sigma):=\sqrt{T_s(\rho,\sigma)},
\end{equation}
where the infimum ranges over finite-dimensional ancillary states
$\gamma$~\cite{Mul22}.  Stabilization improves the monotonicity properties
of the SWAP cost, but it leaves the commuting obstruction unchanged.

\begin{proposition}[Stabilization is trivial on commuting pairs]
\label{prop:stabilization-commuting-pairs}
    If $\rho$ and $\sigma$ commute, then for every ancillary state $\gamma$,
    \begin{equation}\label{eq:commuting-ancilla-invariance}
        T(\rho\otimes\gamma,\sigma\otimes\gamma)=T(\rho,\sigma).
    \end{equation}
    Consequently,
    \[
        T_s(\rho,\sigma)=T(\rho,\sigma),
        \qquad
        W_s(\rho,\sigma)=W(\rho,\sigma).
    \]
\end{proposition}

\begin{proof}
    The SWAP cost is invariant under applying the same unitary to both
    marginals.  We may therefore diagonalize $\rho$ and $\sigma$
    simultaneously and diagonalize $\gamma$ independently.  Write the
    resulting states as $\diag p$, $\diag q$, and $\diag r$.

    For a classical coupling $\pi$, define
    \begin{equation}\label{eq:phi-functional}
        \Phi(\pi):=\sum_{i,j}\sqrt{\pi_{ij}\pi_{ji}}.
    \end{equation}
    The unweighted case of Lemma~\ref{lem:diagonal-state-reduction} is
    equivalent to
    \begin{equation}\label{eq:swap-cost-phi-formula}
        T(\diag p,\diag q)
        =\frac12\left(
        1-\max_{\pi\in\Gamma_{\mathrm{cl}}(p,q)}\Phi(\pi)
        \right).
    \end{equation}

    Let $\Pi$ be a coupling of $p\otimes r$ and $q\otimes r$, with indices
    $(i,a)$ and $(j,b)$, and coarse-grain the ancillary indices by
    \[
        \overline\pi_{ij}
        :=\sum_{a,b}\Pi_{(i,a),(j,b)}.
    \]
    Then $\overline\pi\in\Gamma_{\mathrm{cl}}(p,q)$, and the
    Cauchy--Schwarz inequality gives
    \begin{align*}
        \Phi(\Pi)
        &=\sum_{i,j,a,b}
          \sqrt{\Pi_{(i,a),(j,b)}\Pi_{(j,b),(i,a)}}\\
        &\leq\sum_{i,j}
          \sqrt{\overline\pi_{ij}\overline\pi_{ji}}
          =\Phi(\overline\pi).
    \end{align*}
    Conversely, any $\pi\in\Gamma_{\mathrm{cl}}(p,q)$ lifts to
    \[
        \Pi_{(i,a),(j,b)}
        :=\pi_{ij}r_a\,\mathbf{1}_{\{a=b\}},
    \]
    which is a coupling of $p\otimes r$ and $q\otimes r$ and satisfies
    $\Phi(\Pi)=\Phi(\pi)$.  The maxima in
    \eqref{eq:swap-cost-phi-formula} are therefore equal before and after
    tensoring with $r$, proving \eqref{eq:commuting-ancilla-invariance}.
\end{proof}

\begin{proof}[Proof of Corollary~\ref{cor:intro-stabilized-no-go}]
    The three states in \eqref{eq:qutrit-triple} commute, so
    Proposition~\ref{prop:stabilization-commuting-pairs} applies.  Setting
    $a=b=c=1$ in
    \eqref{eq:cost-rho-sigma}--\eqref{eq:cost-sigma-tau} gives
    \begin{align*}
        T_s(\rho,\sigma)&=\frac{4-\sqrt{10}}8,\\
        T_s(\rho,\tau)&=\frac3{16},\\
        T_s(\sigma,\tau)&=\frac{3-2\sqrt2}{16}.
    \end{align*}
    The square-root violation in
    Proposition~\ref{prop:weighted-qutrit-counterexample} now proves the
    claim.  This transfers Miller's SWAP-cost counterexample~\cite{Mil26}
    to the stabilized cost.
\end{proof}

\section{Fixed-cost no-go for channel-induced couplings}
\label{sec:channel-induced-no-go}

The channel-induced formulation fails at an earlier metric axiom.  The
first step identifies global nonnegativity of the optimized cost with
positivity of the cost operator.

\begin{proposition}[Global nonnegativity]\label{prop:channel-global-positivity}
    The inequality
    \[
        \widetilde T_C(\rho,\sigma)\geq0
        \quad\text{for all states }\rho,\sigma
    \]
    holds if and only if $C\succeq0$.
\end{proposition}

\begin{proof}
    If $C\succeq0$, every feasible expectation is nonnegative.  Conversely,
    let $\omega$ be an arbitrary bipartite state and set
    \[
        \rho=(\operatorname{Tr}_B\omega)^{\top},
        \qquad
        \sigma=\operatorname{Tr}_A\omega.
    \]
    Then $\omega\in\widetilde\Gamma(\rho,\sigma)$.  Global nonnegativity of
    the minimum implies $\operatorname{Tr}(C\omega)\geq0$.  Since this holds
    for every bipartite state, $C\succeq0$.
\end{proof}

The next linear-algebra fact will turn zero self-cost on pure states into a
kernel condition on the whole bipartite space.  Here
$\ket{\overline\psi}$ denotes the entrywise conjugate of $\ket\psi$ in the
fixed basis as introduced in Section~\ref{sec:preliminaries}.

\begin{lemma}[Spanning by conjugate-diagonal tensors]
\label{lem:conjugate-diagonal-spanning}
    \begin{equation}\label{eq:conjugate-diagonal-spanning}
        \operatorname{span}_{\mathbb{C}}
        \left\{\ket{\overline\psi}\otimes\ket\psi:
        \psi\in\calH\right\}
        =\calH\otimes\calH.
    \end{equation}
\end{lemma}

\begin{proof}
    Identify the first factor with the conjugate Hilbert space, and write
    \[
        Q(z):=\ket{\overline z}\otimes\ket z.
    \]
    For arbitrary $x,y\in\calH$, conjugation is conjugate-linear, so
    \[
        \ket{\overline{x+\mathrm{i}^k y}}
        =\ket{\overline x}+(-\mathrm{i})^k\ket{\overline y}.
    \]
    Consequently,
    \begin{align*}
        Q(x+\mathrm{i}^k y)
        &={}
        \ket{\overline x}\otimes\ket x
        +\mathrm{i}^k\ket{\overline x}\otimes\ket y\\
        &\quad
        +(-\mathrm{i})^k\ket{\overline y}\otimes\ket x
        +\ket{\overline y}\otimes\ket y.
    \end{align*}
    Using
    \[
        \sum_{k=0}^3(-\mathrm{i})^k=0,
        \qquad
        \sum_{k=0}^3(-\mathrm{i})^k\mathrm{i}^k=4,
        \qquad
        \sum_{k=0}^3(-1)^k=0,
    \]
    we obtain the polarization identity
    \begin{equation}\label{eq:conjugate-tensor-polarization}
        \ket{\overline x}\otimes\ket y
        =\frac14\sum_{k=0}^3(-\mathrm{i})^k
        Q(x+\mathrm{i}^k y).
    \end{equation}
    Now let $\{e_i\}_{i=1}^d$ be a basis of $\calH$.  The tensors
    $\ket{\overline{e_i}}\otimes\ket{e_j}$ form a basis of
    $\calH\otimes\calH$, and \eqref{eq:conjugate-tensor-polarization}
    expresses each of them as a linear combination of vectors of the form
    $Q(\psi)$.  Hence the latter vectors span $\calH\otimes\calH$.
\end{proof}

\begin{proposition}[Zero self-cost]\label{prop:channel-zero-self-cost}
    Suppose $C\succeq0$ and
    $\widetilde T_C(\rho,\rho)=0$ for every state $\rho$.  Then $C=0$.
\end{proposition}

\begin{proof}
    Take $\rho=\ketbra{\psi}{\psi}$ pure.  A bipartite state with a pure
    marginal is a product with that marginal.  Hence the only element of
    $\widetilde\Gamma(\rho,\rho)$ is
    \[
        \ketbra{\overline\psi\otimes\psi}
                {\overline\psi\otimes\psi}.
    \]
    Zero self-cost gives
    \[
        \bra{\overline\psi\otimes\psi}
        C
        \ket{\overline\psi\otimes\psi}=0.
    \]
    Positivity of $C$ implies
    $C\ket{\overline\psi\otimes\psi}=0$ for every $\psi$.  These vectors
    span $\calH\otimes\calH$ by
    Lemma~\ref{lem:conjugate-diagonal-spanning}, so $C=0$.
\end{proof}

\begin{proof}[Proof of Theorem~\ref{thm:intro-channel-no-go}]
    Proposition~\ref{prop:channel-global-positivity} gives $C\succeq0$.
    Proposition~\ref{prop:channel-zero-self-cost} then gives $C=0$.  For
    $d\geq2$, the resulting optimized cost vanishes for every pair of states
    and therefore cannot separate points.
\end{proof}

\noindent\emph{Comparison.}
The two formulations fail for different reasons.  In the standard
formulation, every admissible cost is positive definite on the antisymmetric
subspace, yet violates the triangle inequality when $d\geq3$.  In the
channel-induced formulation, global nonnegativity and zero self-cost force
$C=0$, already ruling out point separation when $d\geq2$.

\section*{Acknowledgements}
The authors used Large Language Models as AI-assisted research and writing tools throughout the
preparation of this manuscript. The tool was used to help brainstorm ideas and explore proof
strategies. Portions of the manuscript text were redrafted or modified with AI assistance across all
sections. All final mathematical claims, algorithms, proofs, citations, and wording were reviewed,
edited, and validated by the authors. The authors assume responsibility for all content of the
paper.

Minbo Gao would like to thank Li Zhou for helpful discussions. Tianshi Yu would like to thank Lihong Zhi for bringing the problem studied in this paper to his attention and for helpful discussions.

{\small
\bibliographystyle{alphaurl}
\bibliography{main}
}

\end{document}